\documentclass[12pt]{article}
\usepackage[centertags]{amsmath}
\usepackage{amssymb}
\usepackage{amsthm}
\newtheorem{defi}{Definition}[section]

\newtheorem{tw}{Theorem}[section]
\newtheorem{corollary}{Corollary}[section]
\newtheorem{lem}{Lemma}[section]

\newtheorem{prop}{Proposition}[section]
\newtheorem*{rem}{Remark}

\def\={\hspace{-3mm}&=&\hspace{-3mm}}

\renewcommand{\baselinestretch}{1,2}
\title{ The initial value problem for ordinary differential equations with infinitely many derivatives }
\author{ P. G\'orka$^1$\footnote{E-mail: P.Gorka@mini.pw.edu.pl}~, H. Prado$^2$\footnote{E-mail: humberto.prado@usach.cl}~,
and E.G. Reyes$^3$\footnote{E-mail: ereyes@fermat.usach.cl} \\
\small{Department of Mathematics and Information Sciences,}\\
\small{Warsaw University of Technology,}\\
\small{Pl. Politechniki 1, 00-661 Warsaw, Poland} \\
\small{$^{2,3}$ Departamento de Matem\'atica y Ciencia de la
Computaci\'on,} \\
\small{ Universidad de Santiago de Chile }\\
\small{Casilla 307 Correo 2, Santiago, Chile }}

\begin{document}

\maketitle

\begin{abstract}
We study existence, uniqueness and regularity of solutions for
ordinary differential equations with infinitely many derivatives such as
(linearized versions of) nonlocal field equations of motion appearing in particle physics, nonlocal
cosmology and string theory. We develop an appropriate Lorentzian functional calculus via Laplace
transform which allows us to interpret rigorously an operator of the form $f(\partial_t)$ on the half
line, in which $f$ is an analytic function. We find the most general solution to the equation
\begin{equation*}
f(\partial_t) \phi = J(t) \; , \; \; \; t \geq 0 \; ,
\end{equation*}
in the space of exponentially bounded functions, and we also
analyze in full detail the delicate issue of the initial value problem. In particular, we state conditions
under which the solution $\phi$ admits a finite number of derivatives, and we prove rigorously
that if an a priori data directly connected with our Lorentzian calculus is specified, then the initial
value problem is well-posed and it requires only a {\em finite number} of initial conditions.
\end{abstract}

\section{Introduction}\label{sec:Introduction}

In this article we consider ordinary differential equations with an infinite number of
derivatives. These equations have appeared recently as field equations of motion in particle physics
\cite{Mo1,Mo2,Mo3}, in string theory \cite{BCK,D,EW,M,Ta,V,VV,VVZ,W} and in gravity and
cosmology \cite{AV,B,BBC,BK1,BK2,BGKM,Moeller}. For instance, an important equation in this class is
\begin{equation} \label{padic-1}
p^{a\, \partial^2_t}\phi = \phi^p \; , \; \; \; \; \; a > 0 \; .
\end{equation}
Equation (\ref{padic-1}) describes the dynamics of the open $p$-adic string for the scalar tachyon
field (see \cite{AV,BK2,D,Moeller,V,VV,VVZ} and references therein) and it can be understood,
at least formally, as an equation in an infinite number of derivatives if we
expand the left hand side as a power series in $\partial^2_t$. This
equation has been studied rigorously via integral equations of
convolution type in \cite{V,VV} (see also \cite{AV,Moeller}), and it has been also noted
that in the limit $p \rightarrow 1$, Equation (\ref{padic-1}) becomes
the {\em local} logarithmic Klein-Gordon equation \cite{BG,GS}.

In this paper we focus on analytic properties of nonlocal
{\em linear} ordinary differential equations of the form
\begin{equation} \label{abs222}
f(\partial_t) \phi = J(t)
\end{equation}
for $t$ in the half line, and we investigate the initial value problem.
We speak indistinctly of ``nonlocal equations'' or
``equations in an infinite number of derivatives'' in what
follows. We remark that linear nonlocal equations appeared in
mathematics already in the 1930's, see for instance \cite{Car} and
further references appearing in \cite{BK1}. More recently, they
have been considered in connection with the modern theory of
pseudo-differential operators, see for instance \cite{Du,HK,Tran}.
It appears to us, however, that a truly fundamental stimulus for
their study has been the realization that linear and nonlinear
nonlocal equations play a crucial role in contemporary physical
theories.

A serious problem for the development of a rigorous theory for nonlocal
equations has been the difficulties inherent in the understanding of the initial
value problem for equations such as (\ref{padic-1}) and (\ref{abs222}). These difficulties are
carefully analyzed in the classical paper \cite{EW} by Eliezer and
Woodard. One important remark appearing in \cite{EW} is that the
natural idea of considering nonlocal ordinary equations
\begin{equation}\label{xx}
F(t,q',q^{''},\cdots , q^{(n)},\cdots)=0
\end{equation}
as ``limits'' of higher order equations of the form
\begin{equation}\label{x}
G(t,q',q^{''},\cdots , q^{(n)})=0
\end{equation}
is bound to be plagued with problems such as the phenomenon of
Ostrogradsky instability known to appear in the (Lagrangian) study
of equations such as (\ref{x}), see \cite{BK1,EW}. Another
interpretative difficulty considered in \cite{EW} is the
frequently stated argument (see for instance \cite{B,Moeller})
that if an $n$th order ordinary differential equation requires $n$
initial conditions, then the ``infinite order'' equation
(\ref{xx}) requires infinitely many initial conditions, and
therefore the solution $q$ would be determined a priori (via power
series) without reference to the actual equation.

Our approach to these problems is, quite simply stated, to avoid
thinking of Equation (\ref{xx}) as a limit of finite order
equations. It seems to us that the Ostrogradsky instability should
be investigated anew in this nonlocal world, and the same point of
view should be taken with respect to the initial value problem.

Our point of view is to emphasize the fundamental role played by the Laplace
transform ---as a well-defined operator between Banach spaces, see
\cite{ABHN,Doetsch}--- in the analysis of Equation (\ref{abs222})
via a ``Lorentzian functional calculus'', and also in the correct
formulation of the initial value problem. Using this approach we
are able to generalize and formalize completely the interesting
arguments advanced in \cite{BK1,BK2}, in which the Laplace
transform is explicitly mentioned as a natural and powerful tool
for the study of nonlocal equations. Moreover, we can prove rigorously that if
an a priori data directly connected with our Lorentzian calculus is specified,
then the initial value problem is well-posed and it requires only a {\em finite number} of initial
conditions.

It is interesting that neither pseudo-differential operators nor
generalized functions (see for instance \cite{Du,HK,Tran}) appear
explicitly in this work: our ``symbol'' $f(s)$ appearing in
Equation (\ref{abs222}) is (in principle) an arbitrary
analytic function which may be beyond the reach of classical tools
(such as the ones appearing in \cite{H}) and moreover, we stress once
again that we wish to solve nonlocal
equations on the semi-axis $t \geq 0$, not on the real line where perhaps
we could use Fourier transforms and/or classical pseudodifferential analysis
\cite{H}. The importance of studying nonlocal equations on a semi-axis has been recently
stressed by Aref'eva and Volovich \cite{AV}: classical versions of Big Bang cosmological
models contain singularities at the beginning of time, and therefore the time variable
appearing in the field equations should vary over a half line. In \cite{AV} also appears an alternative
analysis of nonlocal equations which applies (for instance) to linearized versions of nonlocal tachyon
fields and $p$-adic strings, but not to linear equations with general analytic symbols.
We also point out that a theory of pseudo-differential operators on the
semi-axis has been recently developed by N. Hayashi and E.I.
Kaikina \cite{HK}, but the precise relation between this theory
and our work remains to be investigated.

This paper is organized as follows. In Section 2 we consider a
rigorous interpretation of nonlocal ordinary differential equations
via a Lorentzian functional calculus founded on the Laplace
transform. For instance, our approach can be applied to the
equation
\begin{equation}
\sqrt{\partial_t^2 - m^2}\,\phi = J(t) \; , \; \; \; \; t \geq 0
\; ,
\end{equation}(see \cite{BK1,Du})
in the space $L_\omega^\infty({\mathbb R}_{+})$, $\,\omega \in
{\mathbb R}$, of exponentially bounded functions. In Section 3 we
investigate (and propose a method for the solution of) initial value problems for linear equations of the
form (\ref{abs222}). This paper generalizes and refines our previous
work \cite{GPR_JMP}, where we considered only entire functions $f$
and we presented a preliminary version of our functional calculus.

\section{Ordinary nonlocal differential equations}

Motivated by the Barnaby-Kamran work \cite{BK1,BK2}, and also by
the interesting papers \cite{Moeller,V,VV}, we are particularly
interested in setting up a rigorous framework in which the initial
value problem for nonlocal equations can be unambiguously
understood. As stated in Section 1, in this work we
consider linear ordinary nonlocal equations of the form
\begin{equation} \label{lor_linear}
  f(\partial_t)\phi(t) - J(t) =0 \; , \; \; \; t > 0\; .
\end{equation}
This equation belongs to a class of equations studied long ago by
R. Carmichael and others \cite{Car,BK1}. In fact, Carmichael in
\cite{Car} states and proves a theorem on existence of solutions
for a class of equations of the form (\ref{lor_linear}) (see also
\cite{BK1}) using power series expansions and techniques from
classical complex analysis. Our approach is certainly based on
this earlier work but we go beyond it in that we
set up a rigorous ``Lorentzian'' functional calculus based on the Laplace transform.
As we will show, this approach effectively brings to
the fore the ``initial conditions'' issue.

We remark once more that solving (\ref{lor_linear}) is not a
simple application of the theory of classical pseudo-differential operators,
since the ``symbols'' $f(s)$ appearing in (\ref{lor_linear}) are in principle
arbitrary (analytic or entire) functions and our aim is to solve nonlocal equations on the semi-axis
$t \geq 0$.

\subsection{Lorentzian functional calculus}

We fix a number $\omega \in \mathbb{R}$. The space
$L^{\infty}_\omega(\mathbb{R_{+}})$ is the Banach space of all
complex-valued and exponentially bounded functions on
$\mathbb{R_{+}}$,
\[
L^{\infty}_\omega(\mathbb{R_{+}}) = \left\{ \phi \in
L^1_{loc}(\mathbb{R_{+}}) : \| \phi \|_{\omega , \infty} =
\mbox{ess} \sup_{t \geq 0} | e^{-\omega t} \phi(t) | < \infty
\right\} \; ,
\]
and the Widder space $C^\infty_W (\omega , \infty)$ is
\[
C^\infty_W (\omega , \infty ) = \left \{ r : (\omega , \infty)
\rightarrow \mathbb{C}\; / \; \| r \|_W = \displaystyle\sup_{n \in
\mathbb{N}_0} \displaystyle \sup_{ s > \omega } \left|
\frac{(s-\omega)^{n+1}}{n !}\, r^{(n)}(s) \right| < \infty \right
\} \; .
\]
As proven in the monograph \cite{ABHN}, functions in $C^\infty_W (\omega , \infty)$
extend analytically to the region $Re(s) > \omega$. Moreover, the Laplace transform
\[
\phi \in L^{\infty}_\omega(\mathbb{R_{+}}) \mapsto {\cal
L}(\phi)(s) = \int_0^\infty e^{-st} \phi (t) dt \in C^\infty_W
(\omega , \infty )
\]
is an isometric isomorphism from
$L^{\infty}_\omega(\mathbb{R_{+}})$ onto $C^\infty_W (\omega ,
\infty)$. We use this isomorphism to define the operator
$f(\partial_t)$ appearing in Equation (\ref{lor_linear}).
Specifically, we construct a functional calculus which allows us to consider $f(\partial_t)$ as
a linear operator on $L^{\infty}_\omega (\mathbb{R_{+}})$.

First of all we make a technical assumption on the function $f$,
generalizing our previous work \cite{GPR_JMP} where we considered
only entire functions: We assume that there exist real numbers
$R_f > 0$ and $\omega_f$ with $\omega_f < R_f$, such that
the function $f$ is analytic in the domain
\[
\{ s : | s | < R_f \} \cup \{s : Re(s) > \omega_f \} \; .
\]

We note that it follows from this assumption that the Taylor series expansion
of $f$ around zero converges absolutely for $| s | < R_f\,$ (see for instance \cite[p. 196]{Hi}).
For ease of reference we say that a function $f$ satisfying the above
condition belongs to the class $\Gamma$.

As motivation, let us perform a formal calculation of $f(\partial_t) \phi\,$:
take $f$ in the class $\Gamma$ and $\phi \in L^{\infty}_\omega(\mathbb{R_{+}})$,
and let us write
\[
 f(\partial_t) = \sum_{n=0}^\infty \frac{f^{(n)}(0)}{n!}\, \partial_t^n \; ,
\]
so that if $\phi$ is smooth, standard properties of the Laplace transform \cite{Doetsch,Sc}
yield
\begin{eqnarray}
\mathcal{L} ( f(\partial_t)\phi ) & = &
\sum_{n=0}^\infty \frac{f^{(n)}(0)}{n!}\, \mathcal{L}(\partial_t^n\phi) \nonumber \\
& = & \sum_{n=0}^\infty \frac{f^{(n)}(0)}{n!}\, (s^n \mathcal{L}(\phi) - s^{n-1}\phi(0) -
                             s^{n-2}\phi'(0) - \dots - \phi^{(n-1)}(0) ) \nonumber \\
& = & \sum_{n=0}^\infty \frac{f^{(n)}(0)}{n!}\, \left(s^n \mathcal{L}(\phi) - \sum_{j=1}^n
                              s^{n-j} \phi^{(j-1)}(0) \right) \nonumber \\
& = & f(s) \mathcal{L}(\phi)(s) - \sum_{n=0}^\infty \sum_{j=1}^n \frac{f^{(n)(0)}}{n!}\,
                              s^{n-j} \phi^{(j-1)}(0) \; . \label{aux1}
\end{eqnarray}
If we define the formal series
\begin{equation}  \label{in_con111}
r(s) =  \sum_{n=1}^\infty \sum_{j=1}^n \frac{f^{(n)}(0)}{n!}\,
d_{j-1} \, s^{n-j} \; ,
\end{equation}
in which $d = \{ d_{j} : j \geq 0 \}$ is a sequence of complex numbers, then we can write (\ref{aux1}) as
\[
 \mathcal{L} ( f(\partial_t)\phi ) = f(s) \mathcal{L}(\phi)(s) - r \; ,
\]
where $d_j = \phi^{(j)}(0)$.

Motivated by this computation, in our previous paper \cite{GPR_JMP} we constructed a
functional calculus using directly series of the form (\ref{in_con111}). Specifically, a slightly
improved version of our previous definition run as follows:

Let $f$ be a function belonging to the class $\Gamma$, set
\begin{equation}
\Lambda = \left\{ d = ( d_{j} )_{ j \geq 0 } : \mbox{the series }
r \mbox{ given by } (\ref{in_con111}) \mbox{ converges for
}Re(s) > \omega_f \right\} \; ,
\end{equation}
and let ${\mathcal R} = \{ r : d \in \Lambda \}$.
Then, the subspace $D$ of $L^{\infty}_{\omega_f}(\mathbb{R_{+}}) \times {\mathcal R}$
consisting of all pairs $(\phi , r)$ such that
\begin{equation} \label{tr00}
\widehat{(\phi,r)} = f\,{\mathcal L}(\phi) - r
\end{equation}
belongs to the Widder space $C^\infty_W (\omega_f , \infty )\,$
is a domain for $f(\partial_t)$ with
\begin{equation} \label{tr11}
f(\partial_t)\, (\phi , r) = {\mathcal L}^{- 1} (\,\,
\widehat{(\phi,r)} \,\,) = {\mathcal L}^{- 1} ( f\,{\mathcal L}(\phi)
- r ) \; .
\end{equation}
We note that with this definition, the operator $f(\partial_t)$ is a linear operator
on $D$, a fact on which we did not insist in \cite{GPR_JMP}.

This definition is reasonable, in the sense that it is directly
related to (\ref{aux1}) and it generalizes the standard case of
polynomial functions $f$. Indeed, if $f(s) = s^N$, then
$f(\partial_t)$ can be interpreted simply as $\partial^{N}_t$. In
this case for each sequence $d = \{ d_{j} : j \geq 0 \}$ the
series $r$ given by (\ref{in_con111}) has only $N$ terms and, if
we take a pair $(\phi , r) \in D$ (as defined above) then
Equations (\ref{tr00}) and (\ref{tr1}) yield
\[
{\mathcal L}(f(\partial_t)\, (\phi,r)) = \widehat{(\phi , r)} =  s^N
{\mathcal L}(\phi) - d_0s^{N-1} - d_1 s^{N-2} - \ldots - d_{N-1}  \; ,
\]
while on the other hand, if $\phi$ is a function of class $C^N$, we
have
\[
{\mathcal L}(\partial_t^N\,\phi) =  s^N {\mathcal L}(\phi) -
\phi(0) s^{N-1} - \phi^{(1)}(0) s^{N-2} - \ldots - \phi^{(N-1)}(0)  \; .
\]
Thus, we can make two observations: {\em (a)} Our definition is compatible with our
intuition that is, $f(\partial_t)\, (\phi,r) = \partial^N_t
\phi$, if $d_j = \phi^{(j)}(0)$ for $j= 1,\dots,N$, but in principle other choices of $d$ are possible.
{\em (b)} The ``remainder'' series $r$ given by (\ref{in_con111}) really encodes all information on initial
values of solutions to equations of the form $f(\partial_t)\phi = J$.

Regretfully, it is technically difficult to operate directly with series $r$ in the study of the initial value problem.
This fact, and the observations above, have led us to considering a more abstract approach:

\begin{defi} \label{def0}
 Let $f$ be a function belonging to the class $\Gamma$ and let $\mathcal{R}$ be the
space of analytic functions on $\mathbb C$.
We consider the subspace $D_f$ of $L^{\infty}_{\omega_f}(\mathbb{R_{+}}) \times \mathcal{R}$
consisting of all the pairs $(\phi , r)$ such that
\begin{equation} \label{tr0}
\widehat{(\phi , r)} = f\,{\mathcal L}(\phi) - r
\end{equation}
belongs to the Widder space $C^\infty_W (\omega_f , \infty )\,$.
The domain of $f(\partial_t)$ as a linear operator on $L^{\infty}_{\omega_f}(\mathbb{R_{+}}) \times \mathcal{R}$
is $D_f$. If $(\phi , r) \in D_f$ then
\begin{equation} \label{tr1}
f(\partial_t)\,(\phi , r) = {\mathcal L}^{- 1} (\,\,\widehat{(\phi,r)} \,\,) =
{\mathcal L}^{- 1} ( f\,{\mathcal L}(\phi)
- r ) \; .
\end{equation}
\end{defi}

Equipped with Definition \ref{def0}, we now study
Equation (\ref{lor_linear}), generalizing our work reported in \cite{GPR_JMP}.
Before doing so, however, let us reconsider the example given after Equation (\ref{tr11}). We would
like to precise the intuitions appearing in the physics literature \cite{B,BBC,BK1,BK2,EW, Moeller}
where operators of the form $f(\partial_t)$ are defined (for appropriate functions $f$) by power
series in $\partial_t$, that is, as ``operators in an infinite number of derivatives''.
Following our previous papers \cite{GPR_JMP,GPR_JDE}, we point out that this idea
can be formalized using a natural extension of the classical theory of analytic vectors \cite{Nelson}:

\begin{defi} \label{av}
Let $A$ be a linear operator on a Banach space $X$, and let $f$ be
a complex-valued function such that $f^{(n)}(0)$ exists for all $n
\geq 0$. We say that $\phi \in X$ is a $f$-analytic vector for $A$
if $\phi$ is in the domain of $A^n$ for all $n \geq 0$ and the series
\[
\sum_{n = 0}^\infty \; \frac{f^{(n)}(0)}{n!}\, A^n(\phi)
\]
defines a vector in $X$.
\end{defi}

We observe that $f$-analytic vectors exist:

\begin{lem}
Fix a function $f$ as in Definition $\ref{av}$. We have:

(a) If $\psi$ is a polynomial function on $\mathbb{R_{+}}$, $\psi$ is an $f$-analytic vector for
$\partial_t$ on $L^{\infty}_\omega(\mathbb{R_{+}})\,$.

(b) If $\psi$ is a $C^\infty$ function on $\mathbb{R_{+}}$ such that
$| e^{-t\omega} \psi^{(n)}(t) | \leq M\, t^n$ for $t < R_f$
and $\psi^{(n)}(t) = 0$ for $t \geq R_f$, then $\psi$ is an analytic $f$-vector for the operator
$\partial_t\,$ on $L^{\infty}_\omega(\mathbb{R_{+}})\,$.
\end{lem}
\begin{proof}
Part (a) is trivial. For (b) we simply note that
\[
\left|e^{-\omega\, t} \sum_{n=0}^\infty \frac{f^{(n)}(0)}{n!} \,
\partial^n_t\,\psi(t) \right|
\leq \sum_{n=0}^\infty \left| \frac{f^{(n)}(0)}{n!} \right| \, |
e^{-\omega\, t}
\partial^n_t\,\psi(t) | \leq \sum_{n=0}^\infty \left|
\frac{f^{(n)}(0)}{n!} \right| \, |t|^n M
\]
and the right hand side of this inequality is finite by
hypothesis.
\end{proof}

We also have the following proposition, connecting $f$-analytic vectors with Definition \ref{def0}:

\begin{prop} \label{bk222}
Assume that $f$ is entire, let $\phi \in L^{\infty}_\omega
(\mathbb{R_{+}})$ be a smooth $f$-analytic vector for $\partial_t$, and
consider the series $r$ given by $(\ref{in_con111})$ with $d_j = \phi^{(j)}(0)$. We
assume that
\begin{equation} \label{marcus}
 | d_j | \leq C\, R^j \; ,
\end{equation}
in which $0 < R < 1$. Then, $(\phi , r) \in D_f$.
\end{prop}
\begin{proof}
Condition (\ref{marcus}) implies that the series $\sum d_{j-1}/s^j$ converges absolutely for
$|s| > R$. It follows from Lemma 2.1 in \cite{GPR_JMP} that the series $r$ is in fact an
entire function, and then Proposition 2.1 of \cite{GPR_JMP} allow us to conclude that
$(\phi , r) \in D_f$.
\end{proof}

We interpret this proposition as stating that, {\em on smooth
$f$-analytic vectors}, the operator $f(\partial_t)$ can indeed be
rigorously understood as an operator in an infinite number of
derivatives, and that this fact is consistent with our Definition
\ref{def0}.

\subsection{Linear nonlocal equations}

In this subsection we solve the nonlocal equation
\begin{equation} \label{lin_gen_0}
f(\partial_t)(\phi , r) = J \;
\end{equation}
using the Lorentzian functional calculus developed above. From now on we will assume
that a function $r \in \mathcal{R}$ has been fixed. We understand Equation (\ref{lin_gen_0}) as an
equation for $\phi \in L^{\infty}_{\omega_f} (\mathbb{R_{+}})$ such that $(\phi , r) \in D_f$.
We simply write $f(\partial_t)\phi = J$ instead of (\ref{lin_gen_0}). First of all, we formalize
what we mean by a solution:

\begin{defi} \label{main_def}
 Let us fix a function $r \in \mathcal{R}$. We say that $\phi \in L^{\infty}_{\omega_f} (\mathbb{R_{+}})$
is a solution to Equation $f(\partial_t)\phi = J$ if and only if
\begin{enumerate}
 \item $\widehat{\phi} = f\,\mathcal{L}(\phi) - r \in C^\infty_W (\omega_f , \infty)\,$; $($ i.e., $(\phi,r) \in D_f$ $)$;
 \item $f(\partial_t)(\phi) = \mathcal{L}^{-1}(f\,\mathcal{L}(\phi) - r) = J\,$.
\end{enumerate}
\end{defi}

Our main theorem on existence and uniqueness of solutions to the linear problem (\ref{lin_gen_0})
is:

\begin{tw} \label{main_thm}
Let us fix a function $f$ in $\Gamma$ and a function $J \in
L^{\infty}_{\omega_f} (\mathbb{R_{+}})$. We assume that the function
$({\cal L}( J ) + r)/f$ is in the Widder space $C^\infty_W (\omega_f , \infty)$. Then, the linear
equation
\begin{equation} \label{lin_gen}
f(\partial_t)\phi = J
\end{equation}
can be uniquely solved for $\phi \in L^{\infty}_\omega (\mathbb{R_{+}})$.
The solution is given by the explicit formula
\begin{equation} \label{sol_lor}
\phi = {\cal L}^{-1} \left( \frac{{\cal L}( J ) + r}{f} \, \right)
\; .
\end{equation}
\end{tw}
\begin{proof}
Since $J \in L^{\infty}_\omega (\mathbb{R_{+}})$, it
follows that $(\phi , r)$, where  $\phi = {\cal L}^{-1} \left(({\cal L}( J ) + r)/f\,\right)\,$,
is in the domain $D_f$ of the operator $f(\partial_t)$: indeed, an easy calculation shows that
$\widehat{\phi}={\cal L}( J )$, which is an element of $C^\infty_W
(\omega_f , \infty)$.  We can then check directly that $\phi$
defined by (\ref{sol_lor}) is a solution of
(\ref{lin_gen}). The isomorphism between
$L^{\infty}_\omega (\mathbb{R_{+}})$ and $C^\infty_W (\omega_f ,
\infty)$ given by the Laplace transform implies uniqueness.
\end{proof}

\begin{rem}
{\rm
We can prove uniqueness using only Definition \ref{main_def} as follows: let us assume that $\phi$ and
$\psi$ are solutions to Equation (\ref{lin_gen}). Then, item 2 of  Definition \ref{main_def} implies
$f\,\mathcal{L}(\phi - \psi) =0$. Set $h = \mathcal{L}(\phi - \psi)$ and suppose that $h(z_0) \neq 0$. By
analyticity, $h(z) \neq 0$ in a neighborhood $U$ of $z_0$. But then $f=0$ in $U$, so that (again by
analyticity) $f$ is identically zero. }
\end{rem}

Theorem \ref{main_thm} can be considered as an abstract version of the
Carmichael theorem on solutions to (\ref{lin_gen}) as it appears
in \cite{Car,BK1}. It is important to notice that the solution
(\ref{sol_lor}) {\em does not need to be either analytic or even
differentiable}: at this stage all we know is that (\ref{sol_lor})
is an element of $L^{\infty}_{\omega_f} (\mathbb{R_{+}})$, so that
in complete generality we cannot even formulate an initial value
problem. In Section 3 we impose further conditions on $J$ and $f$
which assure us that (\ref{sol_lor}) is smooth at $s = 0$, and we
use them to study the initial value problem for (\ref{lin_gen}). Previous
discussions on the initial value problem appearing
in the physics literature are, for instance, \cite{B,BBC,BK1,EW,Moeller}.

We recall that we have fixed an analytic function $r$. Now we
assume the growth condition
\begin{equation}
\Big|\frac{r(s)}{f(s)}\Big|\leq \displaystyle\frac{C}{|s|^p}
\label{growth0}
\end{equation}
for all $|s|$ sufficiently large and some $p > 0$, and we examine
three special cases of Theorem \ref{main_thm}: (a) the function
$r/f$ has no poles; (b) the function $r/f$ has a finite number of
poles; (c) the function $r/f$ has an infinite number of poles.

\begin{corollary} \label{ic}
Assume that the hypotheses of Theorem $\ref{main_thm}$ hold, that
${\cal L}( J )/f$ is in $C^\infty_W(\omega_f , \infty)$, and that
$r/f$ is an entire function such that $(\ref{growth0})$ holds.
Then, solution $(\ref{sol_lor})$ to Equation $(\ref{lin_gen})$ is
simply $\phi = {\cal L}^{-1} \left( \frac{{\cal L}( J )}{f} \,
\right)$. In addition, if $\,1/f$ belongs to $C^\infty_W (\omega_f
, \infty)$, then the solution $(\ref{sol_lor})$ can be written as
a convolution, $\phi = q \ast J$, in which $q$ is uniquely
determined by ${\cal L}(q) = 1/f$.
\end{corollary}

\begin{proof}
The growth condition (\ref{growth0}) implies that there exists $M
> 0$ such that $\displaystyle \Big|\frac{r(s)}{f(s)}\Big|$ is bounded for $|s|>
M$. On the other hand, the function $\displaystyle
 \Big|\frac{r(s)}{f(s)}\Big|$ is bounded for $|s| \leq M$ simply by
continuity. Thus, $r(s)/f(s)$ is an entire function with bounded
module, and therefore $r(s)/f(s) = C_0$ for some constant complex
number. But then (\ref{growth0}) implies that $C_0 =0$, and the
result follows from the general formula (\ref{sol_lor}).
\end{proof}

\begin{corollary} \label{ic1}
Assume that  the hypotheses of\,\, Theorem $\ref{main_thm}$ hold,
and that $\displaystyle \frac{r}{f}$\, has a finite number of
poles $\omega_i$ ($i=1, \dots, N$) of order $\,r_i$ to the left of
$Re(s)=\omega_f$. Suppose also that ${\cal L}( J )/f$ is in
$C^\infty_W(\omega_f , \infty)$, and that the growth condition
$\displaystyle\Big|\frac{r(s)}{f(s)}\Big|\leq
\displaystyle\frac{C}{|s|^p}$\, holds for all $|s|$ sufficiently
large and some $p > 0$. Then, the solution $(\ref{sol_lor})$ can
be written in the form
\begin{equation} \label{car}
\phi(t) = \frac{1}{2 \pi i} \int_{\omega_f - i \infty}^{\omega_f +
i\infty} e^{s\, t} \left( \frac{{\cal L}( J )}{f} \right) \, ds +
\sum_{i=1}^N P_i(t)\,e^{\omega_i t} \; ,
\end{equation}
in which $P_i(t)$ are polynomials of degree $\,r_i - 1$.
\end{corollary}

\begin{proof}
We first notice that the quotient $r/f\in C^\infty_W(\omega_f ,
\infty)$, since $({\cal L}( J ) + r)/f$ and ${\cal L}( J )/f$
are in $C^\infty_W(\omega_f , \infty)$, by assumption. We compute
using the general solution (\ref{sol_lor}) and the inversion
formula for the Laplace transform \cite{Doetsch,Sc}:
\begin{equation}\label{inversion}
\phi(t) =
 \frac{1}{2 \pi
i} \int_{\omega_f - i \infty}^{\omega_f + i\infty} e^{s\, t} \frac{{\cal
L}( J )}{f} ds  + \frac{1}{2 \pi i} \int_{\omega_f - i
\infty}^{\omega_f + i\infty} e^{s\, t} \frac{r}{f} \, ds \; .
\end{equation}
Our growth condition on $r/f$ assures us that the assumptions for
the evaluation of the inversion formula via the calculus of
residues,  \cite[Section 26]{Doetsch}, are verified. Thus, we can
calculate the last integral in the right hand side of
(\ref{inversion}) as
$$
\frac{1}{2 \pi i} \int_{\omega_f - i \infty}^{\omega_f + i\infty} e^{s\,
t} \frac{r}{f} \, ds=\sum_{i=1}^{N}res_{\,i}(t),
$$
in which $res_{\,i}(t)$ denotes the residue of $\displaystyle
\frac{r(s)}{f(s)}$ at $\omega_i$. In order to compute
$res_{\,i}(t)$ we consider the Laurent
expansion around the pole $\omega_i$, that is
\begin{equation} \label{pole-1}
\frac{r(s)}{f(s)}=\frac{a_{1,i}}{(s-\omega_i)}
+\frac{a_{2,i}}{(s-\omega_i)^{2}}+\cdots+\frac{a_{r_i,i}}{(s-\omega_i)^{r_{i}}}+h_i(s)
\end{equation}
where $h_i$ is an analytic function inside a closed curve around
$\omega_i\,$. We multiply (\ref{pole-1}) by $e^{ts}/2\pi i$ and use
Cauchy's integral formula. We obtain (cf. \cite[{\em loc.
cit.}]{Doetsch})
\[
res_{\,i}(t) = P_i(t)\;e^{\omega_i t} \; ,
\]
where $P_i(t)$ is the polynomial of degree $r_i - 1$ given by
\begin{equation}
 P_i(t) = a_{1,i} + a_{2,i}\, \frac{t}{1!} + \cdots + a_{r_i,i}\, \frac{t^{r_i - 1}}{(r_i -1)!} \; .
\end{equation}
\end{proof}

The proof of this corollary is essentially in our previous paper
\cite{GPR_JMP} in the case $f$ entire. We present it here again
because the formula (\ref{car}) for the solution $\phi(t)$ is
precisely Carmichael's formula appearing in \cite{Car}, as quoted
in \cite{BK1}, and also because explicit formulae such as
(\ref{car}) are crucial for the study of the initial value problem
we carry out in Section 3.

The last special case of Theorem \ref{main_thm} is when the
function $r(s)/f(s)$ has an infinite number of isolated poles
$\omega_i$ located to the left of $Re(s) = \omega_f$, and such that
$|\omega_0| \leq |\omega_1| \leq |\omega_2| \leq \cdots$.

\begin{corollary} \label{ic2}
Assume that  the hypotheses of\,\, Theorem $\ref{main_thm}$ hold,
and that $\displaystyle \frac{r}{f}$\, has an infinite number of
poles $\omega_i$ of order $\,r_i$ to the left of
$Re(s)=\omega_f\,$ satisfying $|\omega_i| \leq |\omega_{i+1}|$
for\, $i \geq 0$. We let $\sigma_i$ be curves in the half-plane
$Re(s) \leq \omega_f$ connecting the points $\omega_f + i
\omega_n$ and $\omega_f - i \omega_n\,$, such that $\sigma_n$
together with the segment of the line $Re(s) = \omega_f$ between
these two points encloses exactly the first $n$ poles of
$r(s)/f(s)$. Suppose that ${\cal L}( J )/f$ is in
$C^\infty_W(\omega_f , \infty)$, that the curves $\sigma_n$ are
chosen so that $\omega_n$ tends to infinity as $n$ tends to
infinity, and that
\[
\lim_{n \rightarrow \infty} \int_{\sigma_n} e^{ts} \frac{r(s)}{f(s)} ds = 0 \; .
\]
Then, the solution $(\ref{sol_lor})$ to the linear equation $(\ref{lin_gen})$ can be
written in the form
\begin{equation} \label{car-1}
\phi(t) = \frac{1}{2 \pi i} \int_{\omega_f - i \infty}^{\omega_f +
i\infty} e^{s\, t} \left( \frac{{\cal L}( J )}{f} \right) \, ds +
\sum_{i=1}^\infty P_i(t)\,e^{\omega_i t} \; ,
\end{equation}
in which $P_i(t)$ are polynomials of degree $\,r_i - 1$.
\end{corollary}
\begin{proof}
This corollary is proven as Corollary \ref{ic1}: it follows from the analysis of the complex
inverse Laplace transform appearing in \cite{Doetsch}, see also \cite[p. 160]{Sc}. In particular,
it is explained in \cite[p. 170]{Doetsch} why the series appearing in (\ref{car-1}) indeed converges
for $t \geq 0$.
\end{proof}

We note that the series appearing in Corollary \ref{ic2} is not necessarily differentiable. For instance,
let us take $f = 1 + e^{a s}$ where $a > 0$, and $J = 1$. We wish to solve the equation
\begin{equation} \label{ex2}
(1 + e^{a\,\partial_t})\,\phi = 1 \; , \quad \quad a> 0 \; .
\end{equation}
We take $r = (2 \phi_0 - 1)/s$, so that indeed $r/f$ has an infinite number of poles, $\omega = 0$
and $\omega_n = \frac{2n-1}{a}\,\pi\,i$, $n=0,\pm 1,\pm 2,\cdots$. The solution $\phi$ to Equation
(\ref{ex2}) is
\begin{eqnarray}
\phi & = & \mathcal{L}^{-1}\left( \frac{\mathcal{L}(J) + r}{f} \right) \nonumber \\
     & = & \mathcal{L}^{-1}\left( \frac{1/s + (2\phi_0 - 1)/s}{ 1 + e^{a s}} \right) \nonumber \\
     & = & 2\phi_0\, \mathcal{L}^{-1}\left( \frac{1}{s( 1 + e^{a s})} \right) \; , \label{ex3}
\end{eqnarray}
and the inverse Laplace transform appearing in (\ref{ex3}) is calculated as indicated in Corollary
\ref{ic2}. The answer is in \cite{Sc}:
\begin{equation} \label{ex4}
\phi(t) = 2 \phi_0 \left[ \frac{1}{2} -
\frac{2}{\pi} \sum_{n=0}^\infty \frac{1}{2n-1} \sin \left( \frac{2n-1}{a}\, \pi\,t \right) \right]\; .
\end{equation}
The function $\phi$ satisfies $\phi(0) = \phi_0$, but it is not possible to take $t$-derivative of the
series (\ref{ex4}) in order to get $\phi'(0)\,$! In fact, for $\phi_0 =1$, the function $\phi(t)$ is the
periodic square-wave function
\[
 \phi(t) = \left\{ \begin{array}{ccc} \phi_0 & & 2n \, a \leq t < (2n+1)\,a \; , \\
                                      0      & & (2n+1) \, a \leq t < 2(n+1)\,a \; ,
                   \end{array}
            \right.
\]
$n=0,1,2, \cdots$, and the series expansion (\ref{ex4}) corresponds to its Fourier series representation.

\section{The initial value problem}

In this section we discuss the initial value problem for equations
of the form
\begin{equation} \label{lin_gen_1}
f(\partial_t) \phi = J \; , \quad \quad t \geq 0 \; ,
\end{equation}
in which $f$ belongs to the class $\Gamma$.

\subsection{Generalized initial conditions}

We note that our abstract formula (\ref{sol_lor}) for the solution
$\phi$ to Equation (\ref{lin_gen_1}) tells us that --expanding the
analytic function $r$ appearing in (\ref{sol_lor}) as a power
series (or considering $r=r_d$ where $r_d$ is the series
(\ref{in_con111}))-- $\phi$ depends in principle on an infinite
number of arbitrary constants. However, this fact {\em does not}
mean that the equation itself is superfluous, as formula
(\ref{sol_lor}) for $\phi$ depends essentially on $f$ and $J$. We
think of $r$ as a sort of ``generalized initial condition":

\begin{defi} \label{gics}
 A generalized initial condition for the equation
\begin{equation} \label{ivp0}
f(\partial_t) \phi = J
\end{equation}
is an analytic function $r_0$ such that $(\phi , r_0) \in D_f$ for some
$\phi \in L_{\omega_f}^\infty ({\mathbb R}_+)$.
A generalized initial value problem is an equation such as $(\ref{ivp0})$ together with a generalized
initial condition $r_0$. A solution to a given generalized initial value problem
$\{ (\ref{ivp0}), r_0 \}$ is a function $\phi$ satisfying the conditions of Definition $\ref{main_def}$
with $r = r_0$.
\end{defi}

Thus, given a generalized initial condition, we find a unique
solution for (\ref{lin_gen_1}) using (\ref{sol_lor}), much in the
same way as given {\em one} initial condition we find a unique
solution to  a first order linear ODE. We remark once more (see
comment after the proof of Theorem \ref{main_thm} and the example
after Corollary \ref{ic2}) that there is no reason to believe that
(for a given $r$) the unique solution (\ref{sol_lor}) to
(\ref{ivp0}) will be analytic: within our general context, we can
only conclude that the solution is an integrable exponentially
bounded function. It follows that classical initial value problems
do not even exist in full generality; Definition \ref{gics} is
what replaces them in the framework of our Lorentzian calculus. In
the next subsection we show that --provided $f$ and $J$ satisfy
some extra technical conditions-- we {\em can} consider classical
initial value problems.

We present an explicit example where the foregoing analysis
applies. In \cite[Section 5]{B}, N. Barnaby considers D-brane
decay in a background de Sitter space-time. Up to some parameters,
the equation of motion is
\begin{equation} \label{neil}
e^{-2\, \square} (\square + 1) \phi = \alpha \phi^2 \; ,
\end{equation}
in which $\alpha$ is a constant. In de Sitter space-time we have
$\square = - \partial^2_t - \beta \partial_t$ for a constant
$\beta$, and so Equation (\ref{neil}) becomes
\begin{equation} \label{neil1}
e^{2(\partial_t^2 + \beta \partial_t)} (\partial_t^2 + \beta
\partial_t - 1) \phi = - \alpha \phi^2 \; .
\end{equation}
Expanding the operator $e^{2(\partial_t^2 + \beta \partial_t)}$
formally as a power series, we see that $\phi_0 = 1/\alpha$ solves
(\ref{neil1}). We linearize about $\phi_0$: if $\phi = \phi_ 0 +
\tau\,\psi$, the small deformation $\psi$ satisfies the equation
\begin{equation} \label{neil2}
e^{2(\partial_t^2 + \beta \partial_t)} (\partial_t^2 + \beta
\partial_t - 1) \psi + 2 \psi = 0 \; .
\end{equation}
{\em This} equation can be solved using our Lorentzian calculus!
Consider the entire function
\[
f(s) = e^{2(s^2 + \beta s)} (s^2 + \beta s - 1) + 2 \; ,
\]
and let $r$ be a generalized initial condition, that is, $r$ is an
analytic function such that $r/f$ belongs to Widder space. The
most general solution to (\ref{neil2}) is
\begin{equation}
\psi = \mathcal{L}^{-1}(r/f) \; .
\end{equation}
Now, we can rewrite (\ref{neil2}) as
\begin{equation} \label{neil3}
(\partial_t^2 + \beta \partial_t - 1) \psi = - 2 e^{2
\partial_t^2}\psi (t - 2 \beta) \; ,
\end{equation}
an equation like the ones considered in \cite{B}. We conclude that
$\psi = \mathcal{L}^{-1}(r/f)$ solves (\ref{neil3}), and we note
that, as in Barnaby's analysis, our solution depends on an
essentially arbitrary function given a priori. However, as we show
below, it is enough to assume a {\em finite set} of a priori data
in order to setup initial value problems depending on a finite
number of classical initial conditions.

\subsection{Classical initial value problems}

The main observation on which this subsection rests is that, {\em
if} we can unravel the abstract formula (\ref{sol_lor}) as in
Corollaries \ref{ic}, \ref{ic1} or \ref{ic2}, we can see that in
fact $r$ itself is not essential. The truly important information
needed for formulating (and solving) initial value problems is
encoded in the pole structure of $r(s)/f(s)$.

\begin{defi} \label{civp}
 A classical initial value problem for a nonlocal equation is an equation
\begin{equation} \label{ivp00}
f(\partial_t) \phi = J
\end{equation}
together with a finite set of conditions
\begin{equation} \label{defic}
\phi(0) = \phi_0\; , \quad \phi'(0)= \phi_1\; , \quad \cdots \; ,
\phi^{(k)}(0)=\phi_k \; .
\end{equation}
A solution to a classical initial value problem
$(\ref{ivp00})$-$(\ref{defic})$ is a pair $(\phi,r_0) \in D_f$
satisfying the conditions of Definition $\ref{main_def}$ with $r =
r_0$ such that $\phi$ is differentiable at zero and
$(\ref{defic})$ holds.
\end{defi}

We already remarked in the previous subsection that classical
initial value problems do not exist in general. On the other hand,
as we anticipated (rather roughly) in \cite{GPR_JMP}, following
Barnaby and Kamran's inspiring paper \cite{BK1}, it {\em is}
possible to pose classical initial value problems if we consider,
in addition to a {\em finite number} of initial conditions, some a
priori data directly related to our Laplace transform-based
functional calculus. Intuitively, following Moeller and Zwiebach
\cite[Section 2.3]{Moeller}, a nonlocal equation such as the
$p$-adic string equation
\begin{equation} \label{padicex}
 e^{1/2 \ln p\, \partial_t^2} \phi = \phi^p
\end{equation}
imposes a large number of non-trivial constraints on the set of
possible initial values. Can we describe a consistent set of initial
conditions? In Corollaries \ref{ic},
\ref{ic1}, and \ref{ic2} there are {\em explicit formulas} for
solutions to nonlocal equations of the form (\ref{ivp00}). We would expect these formulas to help us in setting up initial
value problems. Now, Corollary \ref{ic} fixes completely the
solution using {\em only} $f$ and $J$, and therefore it leaves no
room for a classical initial value problem. On the other hand,
formula (\ref{car-1}) of Corollary \ref{ic2} depends on an
infinite number of parameters and, as the example after Corollary
\ref{ic2} shows, this fact implies that we cannot insure
differentiability of the solution (and hence existence of initial value problems) using {\em only} conditions on $f$ and
$J$. But, our explicit formula (\ref{car}) tells us that a
solution for the linear equation (\ref{ivp00}) is uniquely
determined by $f$, $J$, and a {\em finite number} of parameters
related to the singularities of the quotient $r/f$. It is
therefore not unreasonable to expect that, by using these finitely
many parameters, we can set up consistent initial value problems,
as conjectured in \cite{Moeller}.

We use the following two technical lemmas on the differentiability of solutions:

\begin{lem} \label{fo}
Assume that $f$ and $J$ are such that
\begin{equation} \label{cond-1}
y^n\displaystyle\left( \frac{\mathcal{L}(J)}{f} \right) (\omega_f+iy) \in
L^1(\mathbb R)
\end{equation}
for each $n=0,\dots,M$, some $M \geq 0$. Then, the function
\begin{equation} \label{fo1}
t \mapsto \frac{1}{2 \pi i} \int_{\omega_f - i \infty}^{\omega_f +
i \infty} e^{s\,t}\, \left( \frac{\mathcal{L}(J)}{f} \right)(s)\,  ds
\end{equation}
is of class $C^M$.
\end{lem}

The proof of Lemma \ref{fo} consists basically in realizing that the
function (\ref{fo1}) can be viewed as the Fourier transform of the
function
\[
 y \mapsto \left( \frac{\mathcal{L}(J)}{f} \right) (\omega_f+iy) \; .
\]

\begin{lem} \label{tech}
Assume that the conditions of Corollary $\ref{ic1}$ hold, and that
$f$ and $J$ satisfy $(\ref{cond-1})$. Then, solution $(\ref{car})$
to the nonlocal equation
\begin{equation} \label{main}
f(\partial_t)\phi(t) = J(t)
\end{equation}
is of class $C^M$, and it satisfies
\begin{equation}  \label{car-40}
\phi^{(n)}(0) =  L_n + \sum_{i=1}^N\, \sum_{k=0}^{n} \left(\begin{array}{c} n \\
k \end{array} \right) \omega_i^k \left. \frac{d^{n-k}}{dt^{n-k}}
\right|_{t=0} P_i(t) \; , n=0,\dots,M \; ,
\end{equation}
for some numbers $L_n\,$.
\end{lem}

\begin{proof}
Indeed, let us consider Equation (\ref{car}) for the solution
$\phi$ to Equation (\ref{main}). Conditions (\ref{cond-1}) implies
that (\ref{car}) defines a $C^M$ function $\phi$ on the semi-axis
$t \geq 0$. The $t$-derivatives $\phi^{(n)}(t)$ of the solution
$\phi(t)$ are given by
\begin{equation} \label{car-2}
\phi^{(n)}(t) =  \frac{d^{n}}{dt^{n}} \left( \frac{1}{2 \pi i}
\int_{\omega - i \infty}^{\omega + i\infty} e^{s t} \left(
\frac{{\cal L}( J )}{f} \right) \, ds \right) + \sum_{i=1}^N\,
e^{\omega_i t} \sum_{k=0}^{n} \left(\begin{array}{c} n \\
k \end{array} \right) \omega_i^k \frac{d^{n-k}}{dt^{n-k}}P_i(t) \;
,
\end{equation}
and so Equation (\ref{car-40}) follows with
\begin{equation} \label{ele}
L_n = \left.\frac{d^{n}}{dt^{n}}\right|_{t=0} \left( \frac{1}{2 \pi i}
\int_{\omega - i \infty}^{\omega + i\infty} e^{s t} \left(
\frac{{\cal L}( J )}{f} \right) \, ds \right)\; .
\end{equation}
\end{proof}

This lemma was used for us already in \cite{GPR_JMP}, but we reproduce it here since it is crucial for
the analysis that follows.

\begin{tw} \label{ivp2}
Let $f$ be a function on the class $\Gamma$, and fix a function $J$ in $L_{\omega_f}^\infty(\mathbb{R}_+)$
such that $\mathcal{L}(J)/f \in C_W^\infty(\omega_f,\infty)$. Fix also a number $N \geq 0$, a finite number
of points $\omega_i$ to the left of $Re(s)=\omega_f$, and (if $N >0$) a finite number of positive integers
$r_i$, $i =1, . . . ,N$. Set $K=\sum_{i=1}^N r_i$ and assume that condition $(\ref{cond-1})$ holds for all
$n=0,\dots,M$, $M \geq K$. Then, generically, given $K$ initial conditions, $\phi_0, \dots, \phi_{K-1}$,
there exists an analytic function $r_0$ such that
\begin{itemize}
\item[$(\alpha)$] $\displaystyle \frac{r_0}{f}$ has a finite number of poles $\omega_i$ of
                  order $r_i$, $i=1,\dots,N$ to the left of $Re(s)=\omega_f\,$;
\item[$(\beta)$] $\displaystyle \frac{\mathcal{L}(J)+r_0}{f} \in C_W^\infty(\omega_f,\infty)\,$;
\item[$(\gamma)$] $\displaystyle  \left| \frac{r_0(s)}{f(s)} \right| \leq \frac{M}{|s|^p}$
                  for some $p \geq 1$ and $|s|$ sufficiently large.
\end{itemize}
Moreover, the unique solution $\phi$ to Equation $(\ref{ivp00})$ given by $(\ref{sol_lor})$ with $r = r_0$
is of class $C^K$ and it satisfies $\phi(0) = \phi_0, \dots, \phi^{(K-1)}(0)=\phi_{K-1}$.
\end{tw}

\begin{proof}
We take $K$ arbitrary numbers $\phi_n$, $n=0,1,\ldots ,K-1$, and, motivated by Lemma \ref{tech}, we set up
the linear system
\begin{equation}  \label{car-4}
\phi_n =  L_n + \sum_{i=1}^N\, \sum_{k=0}^{n} \left(\begin{array}{c} n \\
k \end{array} \right) \omega_i^k \left. \frac{d^{n-k}}{dt^{n-k}}
\right|_{t=0} P_i(t) \; , \; \; \; n = 0 \dots K-1\; ,
\end{equation}
in which $$L_n = \left.\frac{d^{n}}{dt^{n}}\right|_{t=0} \left( \frac{1}{2 \pi i}
\int_{\omega - i \infty}^{\omega + i\infty} e^{s t} \left(
\frac{{\cal L}( J )}{f} \right) \, ds \right)\; ,$$
and
$$
P_i(t) = a_{1,i} + a_{2,i}\, \frac{t}{1!} + \cdots + a_{r_i,i}\, \frac{t^{r_i - 1}}{(r_i -1)!}
$$
are polynomials to be determined. Condition (\ref{cond-1}) guarantees us that the numbers $L_n$ are well
defined. System (\ref{car-4}) is a linear system for the coefficients of the polynomials $P_i(t)$ which
(generically, depending on the points $\omega_i$) can be solved uniquely in terms of the arbitrary data
$\phi_n\,$. These polynomials allow us to construct the solution to Equation (\ref{main})\,:

We set
\begin{equation} \label{l-1}
r_0(s) = f(s)\, {\mathcal L} \left( \sum_{i=1}^N P_i(t) e^{\omega_i t} \right)(s) \; .
\end{equation}
We have the identity
\begin{equation} \label{l-2}
 {\mathcal L}^{-1}(r_0/f) = \sum_{i=1}^N P_i(t) e^{\omega_i t} \; ,
\end{equation}
and we easily conclude that $r_0/f$ satisfies the conditions ($\alpha$), ($\beta$), and ($\gamma$)
appearing in the enunciate of the theorem. Now we define
\begin{equation} \label{car-7}
\phi(t) = \frac{1}{2 \pi i} \int_{\omega - i \infty}^{\omega +
i\infty} e^{s t} \left( \frac{{\cal L}( J )}{f} \right) \, ds +
\sum_{i=1}^N P_i(t)\,e^{\omega_i\, t} \; ,
\end{equation}
and we claim that this function is the solution to Equation (\ref{main}). In fact,
the foregoing analysis implies that
\[
 \phi(t) = \frac{1}{2 \pi i} \int_{\omega - i \infty}^{\omega +
i\infty} e^{s t} \left( \frac{{\cal L}( J )}{f} \right) \, ds + {\mathcal L}^{-1}(r_0/f) \; ,
\]
and this is precisely the unique solution to (\ref{main}) appearing in Corollary \ref{ic1} for $r=r_0$.

Now we show that this solution satisfies $\phi^{(n)}(0) = \phi_n$ for $n = 0 , \dots, K-1$.
Indeed, condition (\ref{cond-1}) tells us that $\phi(t)$ is at least of class $C^K$ and clearly
\begin{equation}  \label{car-9}
\phi^{(n)}(0) =  L_n + \sum_{i=1}^N\, \sum_{k=0}^{n} \left(\begin{array}{c} n \\
k \end{array} \right) \omega_i^k \left. \frac{d^{n-k}}{dt^{n-k}}
\right|_{t=0} P_i(t)
\end{equation}
in which
\begin{equation} \label{ele-9}
L_n = \left.\frac{d^{n}}{dt^{n}}\right|_{t=0} \left( \frac{1}{2 \pi i}
\int_{\omega - i \infty}^{\omega + i\infty} e^{s t} \left(
\frac{{\cal L}( J )}{f} \right) \, ds \right)\; .
\end{equation}
Comparing (\ref{car-4}) and (\ref{car-9}) we obtain $\phi^{(n)}(0) = \phi_n$, $n = 0, \dots, K-1$.
Thus, we can freely chose the first $K$ derivatives $\phi^{(n)}(0)$,
$n=0, \dots , K-1$.  On the other hand, from $n = K$ onward,
Equation (\ref{car-9}) for the derivatives of $\phi(t)$ and the foregoing analysis imply
that the values of $\phi^{(n)}(0)$, $n \geq K$, are completely determined by the $K$ initial
conditions $\phi_n$.
\end{proof}

We remark that the above proof shows that the points $\omega_i$ become the poles of the quotient $r_0/f$,
and that the numbers $r_i$ are their respective orders. The proof also shows that it is essential to give
{\em a priori} a finite number of points $\omega_i$ in order to have classical initial value problems. If
no points $\omega_i$ are present, the solution to the nonlocal equation (\ref{ivp00}) is simply
\[
\phi = \mathcal{L}^{-1}(\mathcal{L}(J)/f)\; ,
\]
a formula which fixes completely (\,for $f$ and $J$ satisfying (\ref{cond-1})\,) the values of all the
derivatives of $\phi$ at zero.

In conclusion, the following is a complete set of data which allows us to set up classical initial value
problems for equations (\ref{ivp00}) in which $f$ and $J$ satisfy (\ref{cond-1}) sufficiently many times:
\begin{equation} \label{total_data}
\left\{ N \geq 0\; ; \; \; \{\omega_i \in \mathbb{C} \}_{1 \leq i \leq N}\; ;
\; \; \{r_i \in \mathbb{Z} \}_{1 \leq i \leq N}\; ;
\; \; \{ \phi^{(n)}(0) = \phi_n \}_{1 \leq i \leq N} \right\} \; .
\end{equation}
With this data we can set up the linear system (\ref{car-4}) and
effectively construct the unique solution (\ref{car-7}) to
Equation (\ref{ivp00}).

\section{Concluding remarks}

We have developed a Lorentzian functional calculus adequate for interpreting nonlocal operators
appearing in models of particle physics, string theory, and gravity. This calculus is directly
related to the initial value problem for nonlocal equations: as seen in Definition \ref{def0},
the interpretation of an operator of the form $f(\partial_t)$ depends on the choice of some a priori
data (the analytic function $r$ of Definition \ref{def0}). This data is considered as a generalized
initial condition for an equation of the form $f(\partial_t)\phi = J(t)$. However, we can be much more
specific. If $f$ and $J$ satisfy some technical conditions (see Lemmas \ref{fo} and \ref{tech}), we can
take derivatives of the solution, at least a finite number of times, and we can consider a classical
initial value problem for nonlocal equations (Definition \ref{civp}). As shown in Theorem \ref{ivp2},
we can solve this classical initial value problem explicitly: starting with a finite
number of initial conditions and a finite set of extra data (see (\ref{total_data}) and Definition
\ref{civp}) we can obtain a unique regular solution to the given nonlocal equation.

It seems to us that the freedom in the choice of a priori data is
an important feature of our approach, potentially of interest for applications. Thus, it does not appear
to be obvious how to generalize Theorem \ref{ivp2} if $N$ is allowed to be infinite: besides the
fact that if we wish to use the formulae appearing in Corollary
\ref{ic2} we must solve an infinite linear system in
order to obtain an appropriate function $r$ (see proof of Theorem \ref{ivp2}), we would
need to give technical conditions on $f$, $J$ {\em and} on the a priori data
(\ref{total_data}) in order to assure that the solution be
differentiable at $t=0$.

\paragraph{Acknowledgements}

We would like to thank two referees for their thoughtful remarks and constructive criticism.
E.G. Reyes is partially supported by FONDECYT grant \#1111042;
H. Prado is partially supported by project DICYT USACH \#041133PC.

\end{document}